\documentclass[fleqn,12pt,twoside]{article}
\usepackage{tipa}

\usepackage[headings]{espcrc1}
\readRCS
$Id: espcrc1.tex,v 1.2 2004/02/24 11:22:11 spepping Exp $
\usepackage{graphicx}
\usepackage[figuresright]{rotating}

\newtheorem{theorem}{Theorem}

\newtheorem{case}{Case}

\newtheorem{corollary}{Corollary}

\newtheorem{lemma}{Lemma}

\newtheorem{remark}{Remark}

\newenvironment{proof}[1][Proof.]{\begin{trivlist}
\item[\hskip \labelsep {\bfseries #1}]}{\end{trivlist}}

\newcommand{\AmS}{{\protect\the\textfont2
  A\kern-.1667em\lower.5ex\hbox{M}\kern-.125emS}}

\usepackage{amsmath}
\usepackage{amsfonts}
\title{On one-sided interval edge colorings of biregular bipartite
graphs}

\author{R.R. Kamalian\address[MCSD]{Institute for Informatics and Automation Problems, National Academy
of Sciences of RA, 0014 Yerevan, Republic of Armenia}%
\thanks{email: rrkamalian@yahoo.com}}

\runtitle{On one-sided interval edge colorings of biregular bipartite
graphs} \runauthor{R.R. Kamalian}

\begin{document}

\maketitle

\begin{abstract}
A proper edge $t$-coloring of a graph $G$ is a coloring of edges of
$G$ with colors $1,2,...,t$ such that all colors are used, and no
two adjacent edges receive the same color. The set of colors of
edges incident with a vertex $x$ is called a spectrum of $x$. An
arbitrary nonempty subset of consecutive integers is called an
interval. We say that a proper edge $t$-coloring of a graph $G$ is
interval in the vertex $x$ if the spectrum of $x$ is an interval. We
say that a proper edge $t$-coloring $\varphi$ of a graph $G$ is
interval on a subset $R_0$ of vertices of $G$, if for an arbitrary
$x\in R_0$, $\varphi$ is interval in $x$. We say that a subset $R$
of vertices of $G$ has an $i$-property if there is a proper edge
$t$-coloring of $G$ which is interval on $R$. If $G$ is a graph, and
a subset $R$ of its vertices has an $i$-property, then the minimum
value of $t$ for which there is a proper edge $t$-coloring of $G$
interval on $R$ is denoted by $w_R(G)$.

In this paper, for some bipartite graphs, we estimate the value of
this parameter in that cases when $R$ coincides with the set of all
vertices of one part of the graph.

\bigskip
\end{abstract}

We consider undirected, finite graphs without loops and multiple
edges. $V(G)$ and $E(G)$ denote the sets of vertices and edges of a
graph $G$, respectively. For any vertex $x\in V(G)$, we denote by
$N_G(x)$ the set of vertices of a graph $G$ adjacent to $x$. The
degree of a vertex $x$ of a graph $G$ is denoted by $d_G(x)$, the
maximum degree of a vertex of $G$ by $\Delta(G)$. For a graph $G$
and an arbitrary subset $V_0\subseteq V(G)$, we denote by $G[V_0]$
the subgraph of $G$ induced  by the subset $V_0$ of its vertices.

Using a notation $G(X,Y,E)$ for a bipartite graph $G$, we mean that
$G$ has a bipartition $(X,Y)$, and $E=E(G)$.

An arbitrary nonempty subset of consecutive integers is called an
interval. An interval with the minimum element $p$ and the maximum
element $q$ is denoted by $[p,q]$.

A function $\varphi:E(G)\rightarrow [1,t]$ is called a proper edge
$t$-coloring of a graph $G$, if all colors are used, and no two
adjacent edges receive the same color.

The minimum $t\in \mathbb{N}$ for which there exists a proper edge
$t$-coloring of a graph $G$ is denoted by $\chi'(G)$ \cite{Vizing2}.

For a graph $G$ and any $t\in[\chi'(G),|E(G)|]$, we denote by
$\alpha(G,t)$ the set of all proper edge $t$-colorings of $G$. Let
$$
\alpha(G)\equiv\bigcup_{t=\chi'(G)}^{|E(G)|}\alpha(G,t).
$$

If $G$ is a graph, $x\in V(G)$, $\varphi\in \alpha(G)$, then let us
set 
\begin{equation*}
S_G(x,\varphi)\equiv\{\varphi(e)/ e\in E(G), e \textrm{ is
incident with } x\}.
\end{equation*}

We say that $\varphi\in \alpha(G)$ is persistent-interval in the
vertex $x_0\in V(G)$ of the graph $G$ iff
$S_G(x_0,\varphi)=[1,d_G(x_0)]$. We say that $\varphi\in \alpha(G)$
is persistent-interval on the set $R_0\subseteq V(G)$ iff $\varphi$
is persistent-interval in $\forall x\in R_0$.

We say that $\varphi\in \alpha(G)$ is interval in the vertex $x_0\in
V(G)$ of the graph $G$ iff $S_G(x_0,\varphi)$ is an interval. We say
that $\varphi\in \alpha(G)$ is interval on the set $R_0\subseteq
V(G)$ iff $\varphi$ is interval in $\forall x\in R_0$.

We say that a subset $R$ of vertices of a graph $G$ has an
$i$-property iff there exists $\varphi\in \alpha(G)$ interval on
$R$; for a subset $R\subseteq V(G)$ with an $i$-property, the
minimum value of $t$ warranting existence of $\varphi\in
\alpha(G,t)$ interval on $R$ is denoted by $w_R(G)$.

Notice that the problem of deciding whether the set of all vertices
of an arbitrary graph has an $i$-property is $NP$-complete
\cite{Oranj6,Asratian18,Diss7}. Unfortunately, even for an
arbitrary bipartite graph (in this case the interest is strengthened
owing to the application of an $i$-property in timetablings
\cite{Cambridge,Diss7}) the problem keeps the complexity of a
general case \cite{Asratian9,Giaro,Sev}. Some positive results
were obtained for graphs of certain classes with numerical or
structural restrictions \cite{Giaro_Diss,Giaro_Kubale,Kubale,Diss7,Preprint5,Axenovich,Petros_DMath,Petros_Vest,Hansen_Dip,Hans_Lot,Yang,Zhao_Chang}. The examples of bipartite graphs whose
sets of vertices have not an $i$-property are given in
\cite{Cambridge,Giaro_Kubale,Jensen_Toft,Sev}.

The subject of this research is a parameter $w_R(G)$ of a bipartite
graph $G=G(X,Y,E)$ in that case when $R$ coincides with the set of
all vertices of one part of $G$ (the exact value of this parameter
for an arbitrary bipartite graph is not known as yet). We obtain an
upper bound of the parameter being discussed for biregular
\cite{Asratian9,Asratian10,Asratian18,Asratian19,Piatkin20}
bipartite graphs, and the exact values of it in the case of the
complete bipartite graph $K_{m,n}$ $(m\in \mathbb{N}, n\in
\mathbb{N})$ as well.

The terms and concepts that we do not define can be found in
\cite{West1}.

First we recall some known results.

\begin{theorem} \cite{Oranj6,Asratian8,Diss7} \label{thm1}
If $R$ is the set of all vertices of one part of an arbitrary
bipartite graph $G=G(X,Y,E)$, then: 1) there exists
$\varphi\in\alpha(G,|E|)$ interval on $R$, 2) for $\forall
t\in[w_R(G),|E|]$, there exists $\psi_t\in\alpha(G,t)$ interval on
$R$.
\end{theorem}

\begin{theorem} \cite{Oranj6,Asratian8,Asratian14} \label{thm2}
Let $G=G(X,Y,E)$ be a bipartite graph. If for $\forall e=(x,y)\in
E$, where $x\in X, y\in Y$, the inequality $d_G(y)\leq d_G(x)$ is
true, then $\exists\varphi\in\alpha(G,\Delta(G))$
persistent-interval on $X$.
\end{theorem}

\begin{corollary}\cite{Oranj6,Asratian8,Asratian14} \label{cor1}
Let $G=G(X,Y,E)$ be a bipartite graph. If $\max_{y\in
Y}d_G(y)\leq\min_{x\in X}d_G(x),$ then
$\exists\varphi\in\alpha(G,\Delta(G))$ persistent-interval on $X$.
\end{corollary}

\begin{remark}
Note that Corollary \ref{cor1} follows from the result of
\cite{GellerHilton}.
\end{remark}

Let $H=H(\mu,\nu)$ be a $(0,1)$-matrix with $\mu$ rows, $\nu$
columns, and with elements $h_{ij}$, $1\leq i\leq\mu$, $1\leq
j\leq\nu$. The $i$-th row of $H$, $i\in[1,\mu]$, is called
collected, iff $h_{ip}=h_{iq}=1$, $t\in[p,q]$ imply $h_{it}=1$, and
the inequality $\sum_{j=1}^\nu h_{ij}\geq 1$ is true. Similarly, the
$j$-th column of $H$, $j\in[1,\nu]$, is called collected, iff
$h_{pj}=h_{qj}=1$, $t\in[p,q]$ imply $h_{tj}=1$, and the inequality
$\sum_{i=1}^\mu h_{ij}\geq 1$ is true. If all rows and all columns
of $H$ are collected, then for $i$-th row of $H$, $i\in[1,\mu]$, we
define the number $\varepsilon(i,H)\equiv\min\{j/h_{ij}=1\}$.

$H$ is called a collected matrix, iff all its rows and all its
columns are collected, $h_{11}=h_{\mu\nu}=1$, and
$\varepsilon(1,H)\leq\varepsilon(2,H)\leq\cdots\leq
\varepsilon(\mu,H)$.

$H$ is called a $b$-regular matrix $(b\in \mathbb{N})$, iff for
$\forall i\in[1,\mu]$, $\sum_{j=1}^\nu h_{ij}=b$. $H$ is called a
$c$-compressed matrix $(c\in \mathbb{N})$, iff for $\forall
j\in[1,\nu]$, $\sum_{i=1}^\mu h_{ij}\leq c$.

\begin{lemma} \cite{Vestnik15} \label{lem1}
If a collected $n$-regular $(n\in \mathbb{N})$ matrix $P=P(m,w)$
with elements $p_{ij}$ $(1\leq i\leq m, 1\leq j\leq w)$ is
$n$-compressed, then $w\geq\big\lceil\frac{m}{n}\big\rceil\cdot n$.
\end{lemma}

\begin{proof}
We use induction on $\big\lceil\frac{m}{n}\big\rceil$.

If $\big\lceil\frac{m}{n}\big\rceil=1$, the statement is trivial.

Now assume that $\big\lceil\frac{m}{n}\big\rceil=\lambda_0\geq 2$,
and the statement is true for all collected $n'$-regular
$n'$-compressed matrixes $P'(m',w')$ with
$\big\lceil\frac{m'}{n'}\big\rceil\leq\lambda_0-1$.

First of all let us prove that $\varepsilon(n+1,P)\geq n+1$. Assume
the contrary: $\varepsilon(n+1,P)\leq n$. Since $P$ is a collected
$n$-regular matrix, we obtain $\sum_{i=1}^m
p_{in}\geq\sum_{i=1}^{n+1} p_{in}\geq n+1$, which is impossible
because $P(m,w)$ is an $n$-compressed matrix. This contradiction
shows that $\varepsilon(n+1,P)\geq n+1$.

Now let us form a new matrix $P'(m-n,w-(\varepsilon(n+1,P)-1))$ by
deleting from the matrix $P$ the elements $p_{ij}$, which satisfy at
least one of the inequalities $i\leq n$, $j\leq
\varepsilon(n+1,P)-1$.

It is not difficult to see that $P'(m-n,w-(\varepsilon(n+1,P)-1))$
is a collected $n$-regular $n$-compressed matrix with
$\big\lceil\frac{m-n}{n}\big\rceil=\lambda_0-1$. By the induction
hypothesis, we have
$$
w-(\varepsilon(n+1,P)-1)\geq\bigg\lceil\frac{m-n}{n}\bigg\rceil\cdot
n,
$$
which means that
$$
w\geq(\lambda_0-1)n+\varepsilon(n+1,P)-1\geq(\lambda_0-1)n+n=
\lambda_0n=\bigg\lceil\frac{m}{n}\bigg\rceil\cdot n.
$$
\end{proof}

Now, for arbitrary positive integers $m,l,n,k$, where $m\geq n$ and
$ml=nk$, let us define the class $Bip(m,l,n,k)$ of biregular
bipartite graphs:
$$
Bip(m,l,n,k)\equiv\left\{G=G(X,Y,E)\bigg/
\begin{array}{lc} |X|=m, |Y|=n, \textrm{ for }  \forall
x\in X, d_G(x)=l, \\ \textrm{ for } \forall y\in Y, d_G(y)=k.
\end{array}
\right\}
$$

\begin{remark} \label{rem1}
Clearly, if $G\in Bip(m,l,n,k)$, then $\chi'(G)=k$.
\end{remark}

\begin{theorem} \label{thm3}
If $G=G(X,Y,E)\in Bip(m,l,n,k)$, then $w_Y(G)=k$, $w_X(G)\leq
l\cdot\big\lceil\frac{m}{l}\big\rceil$.
\end{theorem}

\begin{proof}
The equality follows from Remark \ref{rem1}. Let us prove the
inequality.

Let $X=\{x_1,\ldots,x_m\}$. For $\forall
r\in[1,\lfloor\frac{m}{l}\rfloor]$, define
$X_r\equiv\{x_{(r-1)l+1},\ldots,x_{rl}\}$. Define
$X_{1+\lfloor\frac{m}{l}\rfloor}\equiv
X\backslash\Big(\bigcup_{i=1}^{\lfloor\frac{m}{l}\rfloor}X_i\Big)$.
For $\forall r\in[1,\lfloor\frac{m}{l}\rfloor]$, define
$Y_r\equiv\bigcup_{x\in X_r}N_G(x)$. Define
$Y_{1+\lfloor\frac{m}{l}\rfloor}\equiv \bigcup_{x\in
X_{1+\lfloor\frac{m}{l}\rfloor}}N_G(x)$. For $\forall
r\in[1,\big\lceil\frac{m}{l}\big\rceil]$, define $G_r\equiv
G[X_r\cup Y_r]$.

Consider the sequence $G_1,G_2,\ldots,G_{\lceil\frac{m}{l}\rceil}$
of subgraphs of the graph $G$. From Corollary \ref{cor1}, we obtain
that for $\forall i\in[1,\lceil\frac{m}{l}\rceil]$, there is
$\varphi_i\in\alpha(G_i,l)$ persistent-interval on $X_i$.

Clearly, for $\forall e\in E(G)$, there exists the unique $\xi(e)$,
satisfying the conditions $\xi(e)\in[1,\lceil\frac{m}{l}\rceil]$ and
$e\in E(G_{\xi(e)})$.

Define a function $\psi:E(G)\rightarrow
[1,l\cdot\big\lceil\frac{m}{l}\big\rceil]$. For an arbitrary $e\in
E(G)$, set $\psi(e)\equiv(\xi(e)-1)\cdot l+\varphi_{\xi(e)}(e)$.

It is not difficult to see that
$\psi\in\alpha(G,l\cdot\big\lceil\frac{m}{l}\big\rceil)$ and $\psi$
is interval on $X$. Hence, $w_X(G)\leq
l\cdot\big\lceil\frac{m}{l}\big\rceil$.
\end{proof}

\begin{theorem} \label{thm4}
Let $R$ be the set of all vertices of one part of the complete
bipartite graph $G=K_{m,n}$, where $m\in \mathbb{N}$, $n\in
\mathbb{N}$. Then
$$
w_R(G)=(m+n-|R|)\cdot\bigg\lceil\frac{|R|}{m+n-|R|}\bigg\rceil.
$$
\end{theorem}

\begin{proof}
Without loss of generality we can assume that $G$ has a bipartition
$(X,Y)$, where $X=\{x_1,\ldots,x_m\}$, $Y=\{y_1,\ldots,y_n\}$, and
$m\geq n$.

\case{1} $R=Y$. In this case the statement follows from Theorem
\ref{thm3}; thus $w_Y(G)=m$.

\case{2} $R=X$.

The inequality $w_X(G)\leq n\cdot\big\lceil\frac{m}{n}\big\rceil$
follows from Theorem \ref{thm3}. Let us prove that $w_X(G)\geq
n\cdot\big\lceil\frac{m}{n}\big\rceil$.

Consider an arbitrary proper edge $w_X(G)$-coloring $\varphi$ of the
graph $G$, which is interval on $X$.

Clearly, without loss of generality, we can assume that
$$
\min(S_G(x_1,\varphi))\leq\min(S_G(x_2,\varphi))\leq\ldots\leq\min(S_G(x_m,\varphi)).
$$

Let us define a $(0,1)$-matrix $P(m,w_X(G))$ with $m$ rows, $w_X(G)$
columns, and with elements $p_{ij}$, $1\leq i\leq m$, $1\leq j\leq
w_X(G)$. For $\forall i\in[1,m]$, and for $\forall j\in[1,w_X(G)]$,
set
$$
p_{ij}= \left\{
\begin{array}{lc}
1, & \textrm{if }
j\in S_G(x_i,\varphi) \\
0, & \textrm{if } j\not\in S_G(x_i,\varphi).
\end{array}
\right.
$$

It is not difficult to see that $P(m,w_X(G))$ is a collected
$n$-regular $n$-compressed matrix. From Lemma \ref{lem1}, we obtain
$w_X(G)\geq n\cdot\big\lceil\frac{m}{n}\big\rceil$.
\end{proof}

From Theorems \ref{thm1} and \ref{thm3}, taking into account the
proof of Case 2 of Theorem \ref{thm4}, we also obtain

\begin{theorem}
If $G\in Bip(m,l,n,k)$, then
\begin{enumerate}
\item for $\forall t\in
\Big[l\cdot\big\lceil\frac{m}{l}\big\rceil,ml\Big]$, there exists
$\varphi_t\in\alpha(G,t)$ interval on $X$,
\item for $\forall t\in [k,nk]$, there exists $\psi_t\in\alpha(G,t)$
interval on $Y$.
\end{enumerate}
\end{theorem}

\end{document}